\documentclass[11pt]{article}
\usepackage[a4paper]{geometry}
\usepackage[english]{babel}

\usepackage{amsmath, accents}
\usepackage{amsfonts}
\usepackage{amstext}
\usepackage{amssymb}
\usepackage{amsthm}
\usepackage{amscd}
\usepackage{mathrsfs}

\usepackage{fancyhdr}
\usepackage[pagebackref,draft=false]{hyperref}
\hypersetup{colorlinks,
linkcolor=myrefcolor,
citecolor=mycitecolor,
urlcolor=myurlcolor}

\usepackage[capitalize]{cleveref}
\usepackage{cite}
\usepackage{caption}
\usepackage{etaremune}

\usepackage{xcolor}
\definecolor{myurlcolor}{rgb}{0,0,0.4}
\definecolor{mycitecolor}{rgb}{0,0.5,0}
\definecolor{myrefcolor}{rgb}{0.5,0,0}
\usepackage{graphicx}
\usepackage{tikz}
\usepackage{tikz-cd}
 \usetikzlibrary{decorations.pathmorphing}

\usepackage{etoolbox}
\usepackage{makeidx}
\usepackage{sectsty}
\usepackage{enumitem} 
\usepackage[]{latexsym}
\usepackage{braket}
\usepackage{caption}
\usepackage[utf8]{inputenx}
\usepackage[T1]{fontenc}
\usepackage{lmodern}
\usepackage{textcomp}
\usepackage{microtype}
\usepackage{totcount}
\usepackage{blindtext}
\newtheorem{theorem}{Theorem}

\newtheorem{proposition}{Proposition}

\newtheorem{lemma}{Lemma}

\newtheorem*{proof*}{Proof}

\newcommand{\be}{\begin{equation}}
\newcommand{\ee}{\end{equation}}
\newcommand{\bea}{\begin{eqnarray}}
\newcommand{\eea}{\end{eqnarray}}

\newcommand{\blue}[1]{\textcolor{blue}{{#1}}}

\newcommand{\ra}{\rightarrow}

\newcommand{\lra}{\longrightarrow}

\newcommand{\hh}{\mathcal{H}}

\newcommand{\valg}{\mathcal{V}(\mathbf{G})}
\newcommand{\valga}{\mathcal{V}(\mathbf{G}_{A})}
\newcommand{\valgxa}{\mathcal{V}(\mathbf{G}_{x_{A}})}
\newcommand{\valgb}{\mathcal{V}(\mathbf{G}_{B})}
\newcommand{\valgab}{\mathcal{V}(\mathbf{G}_{A}\times\mathbf{G}_{B})}

\title{Schr\"{o}dinger's problem with cats: measurements and states in the Groupoid Picture of Quantum Mechanics}

\author{F. M. Ciaglia$^{1,5}$  \href{https://orcid.org/0000-0002-8987-1181}{\includegraphics[scale=0.7]{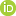}}, F. Di Cosmo$^{2,3,6}$ \href{https://orcid.org/0000-0003-0256-5913}{\includegraphics[scale=0.7]{ORCID.png}}, A. Ibort$^{2,3,7}$\href{https://orcid.org/0000-0002-0580-5858}{\includegraphics[scale=0.7]{ORCID.png}}, G. Marmo$^{4,8}$\href{https://orcid.org/0000-0003-2662-2193}{\includegraphics[scale=0.7]{ORCID.png}}\\
\footnotesize{$^{1}$\textit{ Max Planck Institute for Mathematics in the Sciences, Leipzig, Germany}} \\
\footnotesize{$^{2}$\textit{ ICMAT, Instituto de Ciencias Matem\'{a}ticas (CSIC-UAM-UC3M-UCM)}}  \\
\footnotesize{$^{3}$\textit{ Depto. de Matem\'aticas, Univ. Carlos III de Madrid, Legan\'es, Madrid, Spain}}  \\
\footnotesize{$^{4}$\textit{ Dipartimento di Fisica ``E. Pancini'', Universit\`a di Napoli Federico II, Napoli, Italy}} \\
\footnotesize{$^{5}$\textit{ e-mail: \texttt{florio.m.ciaglia[at]gmail.com}}}, \\
\footnotesize{ $^{6}$\textit{ e-mail: \texttt{fabio[at]}}} \\ 
\footnotesize{ $^{7}$\textit{ e-mail: \texttt{albertoi[at]math.uc3m.es}}} \\ 
\footnotesize{$^{8}$\textit{ e-mail: \texttt{marmo[at]na.infn.it}}}}

\date{}

\begin{document}

\maketitle

\begin{abstract}
 Schr\"{o}dinger's famous {\itshape Gedankenexperiment} involving a cat is used as a motivation to discuss the evolution of states of the composition of classical and quantum systems in the groupoid formalism for physical theories introduced recently.  
%%%%
It is shown that the notion of classical system, in the sense of Birkhoff and von Neumann, is equivalent, in the case of systems with a countable number of outputs, to a totally disconnected groupoid with Abelian von Neumann algebra.
%%%
In accordance with Raggio's theorem, the impossibility of evolving the product state of a composite system made up of a classical and a quantum one into an entangled state by means of a unitary evolution which is internal to the composite system itself is proved in the groupoid formalism.
\end{abstract}

\tableofcontents

\thispagestyle{fancy}

\section*{Remark}

A \href{https://www.mdpi.com/1099-4300/22/11/1292}{slightly different version} this work will  has  appeared in Entropy. 
%%%%
The main difference is that in the version published in Entropy no mention is made of Schr\"{o}dinger's cat Gedankenexperiment. 
%%%%
This choice was made to comply with the suggestions of one Referee and of the Editors.
%%%
However, we also decided to make available on arXiv  the version of the paper in which we motivate the connection between our work and the Schr\"{o}dinger's cat Gedankenexperiment, when the latter is considered as referring to entanglement production during the unitary evolution of  the composition of a quantum system with a classical system (e.g., the cat in Schr\"{o}dinger's original paper), and not as referring to the issue of the Quantum Measurement Problem. Of course, comments and suggestions are more than welcome!

\section{Introduction}\label{sec.1:Introduction}

Erwin Schr\"{o}dinger shared with Einstein an enormous puzzlement about the implications of the laws that were being uncovered in the investigation of atomic processes.
%%%
In their {\itshape Gedankenexperiment} \cite{E-P-R-1935}, Einstein, Podolski and Rosen showed the conflicting relation between ``Elements of Physical Reality'' and the notions of separability and independence in Quantum Mechanics (see the recent analysis of such situation in the recent paper \cite{C-DC-I-M-2020}).
%%%
Schr\"{o}dinger showed his bewilderment in a series of reflections summarized by his famous experiment involving a macroscopic body (a cat) and a quantum system \cite{Schrodinger-1935}, where he argued about the conflict between ``common sense'' and what we now refer to as an entangled state between a cat and some radioactive material.
%%%% 

To be more precise, Schr\"{o}dinger considered a physical system consisting of a cat, described by a quantum state defined as a (normalized) vector in a Hilbert space $\mathcal{H}_{cat}$. 
%%%
For the sake of simplicity, we will assume that there are two states, say $|A\rangle$ and $|D\rangle$ in   $\mathcal{H}_{cat}$, representing the fact that the cat may be alive or dead, respectively. 
%%%%
The cat is put into a box containing a small quantity of radioactive material, so small that it could emit just a particle during a period of time long enough. 
%%%%
Inside the box there is a detector which is activated by the emission of one particle by the radioactive material. 
%%%
When the detectors is triggered, it activates a device that releases a toxin that kills the cat almost immediately.
%%%
We will describe the quantum state of the radioactive material as either the vector $|+\rangle$ or $|-\rangle$ in some Hilbert space $\mathcal{H}_R$, depending if it is in an excited state or it has emitted a particle, respectively. 
%%%
Consequently, according to most of the standard interpretations of Quantum Mechanics, the Hilbert space of the composite system is given by the tensor product $\mathcal{H}_{cat}\otimes\mathcal{H}_{R}$.
%%%%

The alive cat is put in the box together with the radioactive material in the excited (non-emitting) state, that is, the initial state of the system is given by
\be\label{eqn: schrodinger's cat state}
|\psi_{0}\rangle\,=\,|A\rangle\,\otimes\,|+\rangle\,,
\ee
and the composite system is left to evolve according to its isolated dynamical evolution.
%%%%
The state of the system at time $t$ is given by
\begin{equation*}
    |\psi_t \rangle = U_t |\psi_0 \rangle\,,
\end{equation*}
where $U_t$ is the unitary operator on $\mathcal{H}_{cat} \otimes \mathcal{H}_R$ determined by the specific interactions between all systems involved, including the dynamics of the radioactive material itself.
%%%%
The dynamical evolution of the system may induce the radioactive material to emit or not, and since the Hilbert space of the total system is a tensor product, this instance reflects in the possibility that the dynamical evolution leads to an entangled state $|\psi_{t}\rangle$ at time $t$, which is described, for instance, by the vector  
\be\label{eqn: alive-dead entangled state}
|\psi_{t} \rangle = \frac{1}{\sqrt{2}}\left( |A\rangle \otimes |+\rangle + |D\rangle \otimes|- \rangle \right) \,,
\ee
and which represents a superposition of the alive cat with the dead cat.
%%%%
Schr\"{o}dinger, of course, considered this possibility to be absurd. 
%%%%
However, later on, such entangled states for composite systems made of macroscopic systems and microscopic ones were actually constructed and provided a further verification of the quantum mechanical nature of physical reality (GHZ state) \cite[Ch. 24]{Auletta-2001}, and contributed to promote entanglement as a foundational basis for modern quantum information technologies. 
%%%%

To resolve the contradiction arising from our intuition about Schr\"{o}dinger's cat and the experimental verification of the existence of entangled states, we must realize that there are some subtle issues concerning the previously described situation that  deserve further analysis. 
%%%
For instance, Schr\"{o}dinger identified the notion of ``macroscopic'' system with the notion  of ``classical'' system, while, today, we know that these two notions need not coincide in general.
%%%%
Then, we must note  that the experimental construction of entangled states is in general not a trivial matter, and this is why entangled states are considered precious resources in the modern theories known as ``resource theories'' \cite{C-G-2019}.
%%%
Therefore, we should rephrase Schr\"{o}dinger's original objection as: ``\textit{Starting with a composite system   made up of a \textbf{classical} and a   \textbf{quantum} part, and which is in a product state, is it possible to build an \textbf{entangled state} by means of a unitary evolution of the system?}''
%%%

In the context of the theory of von Neumann algebras, Raggio's theorem \cite{Raggio-1981} clearly shows that this is not true.
%%%%
In this contribution, we will analyse this question in the groupoid  formalism developed recently in \cite{C-DC-I-M-2020,C-DC-I-M-02-2020,C-I-M-2018,C-I-M-02-2019,C-I-M-03-2019,C-I-M-05-2019}.
%%%
We stress that we do not address the aspects of the Schr\"{o}dinger cat paradox which are related with the quantum measurement problem \cite{Maudlin-1995}, but leave this investigation to future works. 
%%%
The result presented here will be that, in accordance with Raggio's theorem, it is actually impossible to build an entangled state out of a product state of a classical and a quantum system.
%%%%
This result will crucially depend on the precise meaning of some of the words involved in the previous discussion.
%%%
Specifically,  we need to carefully define what we mean by ``classical'' and ``quantum'' system, their ``composition'', and the ``entangled states'' of the composite system.
%%%
In the following sections, we will elaborate on the meaning of these notions in the groupoid formalism, and we will also present an analysis of the logic approach of Birkhoff and von Neumann, in order to show that the notion of classical system introduced in this formalism coincides with the notion of classical system in the groupoid formalism and is connected with the commutativity of the groupoid algebra of the system.
%%%
Moreover, the formulation of the problem in the groupoid formalism will allow to present an example which escapes Raggio's theorem in the sense that, even if both algebras of the systems are non-Abelian, there are separable states which are sent to separable states by all unitary dynamics of the composite system.

%%%

%%%%%%%%%%%%
%%%%%%%%%%%%

\section{Birkhoff-von Neumann's algebra of propositions}\label{sec:propositions}

In this section, we will briefly recall the propositional calculus of Birkhoff and von Neumann \cite{B-vN-1936} \cite[Ch. 5]{Jauch-1952}.
%%%
This will serve as a mean to investigate the relation between this formalism and the groupoid formalism in the next section, and to show that the notion of classical system, which will be of crucial importance in the analysis of Schr\"{o}dinger's {\itshape Gedankenexperiment}, ``coincides'' in the two formalisms in the sense explained in Theorem \ref{thm: equivalence between classical systems in the groupoid formalism and the propositional calculus}.
%%%

Einstein and Schr\"{o}dinger were not the only ones having difficulties with the foundational aspects of Quantum Mechanics. Indeed, also J. von Neumann  had his  own.
%%%%
His work had an enormous influence because, while providing a solid mathematical background for Quantum Mechanics (and setting the foundations of the theory of operators algebras), with the aim of clarifying the situation \cite{von-Neumann-1955}, he introduced the idealization of two basic physical operations: measurements and compositions. 
%%%%
Without entering here on a discussion on von Neumann's influence on both notions and our interpretation of them, we would like to point out that, even prior to those, the notion of state worried von Neumann the most.
%%%%
Therefore, we will concentrate on this concept in the present analysis, not because the other two notions are germane to this discussion, but because it will clarify some of the foundational aspects of the groupoid picture of Quantum Mechanics that were left open in \cite{C-I-M-05-2019}. 
%%%%

In order to provide a sound mathematical foundation for Quantum Mechanics  beyond the Hilbert space formalism, J. von Neumann together with G. Birkhoff   investigated the logical structure of Quantum Mechanics.
%%%
They took an operational point of view that could be summarised in the statement that a system is described by the results of the interrogations we put on it. 
%%%
Thus, they devised a logical system whose  propositions $\mathcal{P}$ represent  experiments performed on the system with just two possible outcomes  (Yes-No experiments); thus, the system could pass the experiment (unchanged), or it does not pass. 
%%%
There is a natural relation between propositions, denoted $ P \subseteq Q $, meaning that if  the system passes $P$ it will also pass $Q$. 
%%%
As we will see later on, this way of thinking about propositional calculus becomes quite close to Schwinger's algebra of selective measurements. 
%%%%

After this, a set of axioms, based on their clarity and physical meaning, was proposed characterizing  the families of propositions $(\mathcal{P}, \subseteq)$ describing physical systems.
%%%
Specifically, these axioms are:
\begin{enumerate}
    \item the relation $\subseteq$ is a partial order;
    \item $\mathcal{P}$ contains two special propositions: $E$, which always allows the passage of the system (tautology), and $\emptyset$, which never permits the passage of the system (absurd); 
    \item the pair $(\mathcal{P}, \, \subseteq )$ is a lattice, i.e., given $P, Q \in \mathcal{P}$, there exists the least upper bound of $P$ and $Q$, denoted as $P\cup Q$, and the greatest lower bound, $P \cap Q$; 
    \item the lattice $(\mathcal{P}, \, \subseteq )$ is orthocomplemented, i.e., for any $P\in \mathcal{P}$, there is another proposition $P'$ such that:
    \begin{enumerate}
        \item $\emptyset' = E$, and $E' = \emptyset$;
        \item $(P')'=P$, $P\cap P'= \emptyset$, $P\subseteq Q \Leftrightarrow Q' \subseteq P'$;
        \item $(P\cup Q)' = P' \cap Q'$.
    \end{enumerate}
\end{enumerate}

Finally, it was argued that lattice $(\mathcal{P}, \subseteq, \cup, \cap, ')$ is  weakly modular, i.e., if $P\subseteq Q$, then $P$ and $Q$ are compatible, which means that they generate a Boolean sublattice of $\mathcal{P}$ (see the discussion in \cite[Ch. 5]{Jauch-1952}). 
%%%
In addition to the previous set of axioms, a Boolean lattice satisfies also the distributive law encoded in the relations
\be\label{eqn: Boolean lattice}
\begin{split}
P\cup ( Q \cap R) & = (P \cup Q) \cap (P \cup R) \\
P \cap (Q \cup R) &= (P \cap Q) \cup (P \cap R).
\end{split}
\ee
%%%
Sometimes the previous list of properties of $\mathcal{P}$ are supplied with an atomicity condition which is not relevant for our discussion.
%%%
In \cite{B-vN-1936}, it was proved that the lattice of propositions of a quantum system determines a projective geometry and propositions can be represented as orthogonal projectors in a Hilbert space, i.e.,  the standard picture of Quantum Mechanics in terms of Hilbert spaces is recovered (see also \cite[Ch. 8]{Jauch-1952}). 
%%%

The propositional calculus discussed so far allows for a clear interpretation of the notion of state of a quantum system where its statistical interpretation becomes apparent. 
%%%%
A state of a quantum system described by a lattice of propositions $\mathcal{P}$ is a map $p\, \colon \, \mathcal{P} \, \rightarrow \, \mathbb{R}_+$ satisfying the following axioms:
\begin{itemize}
    \item $p(\emptyset) = 0$ and $p(E)=1$;
    \item $p(P\cup Q) = p(P) + p(Q)$, if $P\cap Q = \emptyset$\,.
\end{itemize}
%%%
If the lattice  is closed under joint of countable proposition these axioms can be extended to their $\sigma$-counterpart.
%%%

Hence, because of the additivity property of the function $p$, a statistical interpretation of the quantity $p(P)$, as in standard probability theory, as the statistical frequency that the system will pass the experiment $P$ in a large series of trials, can be introduced. 
%%%
This constitutes, actually, the formalization of the notion of ``Gesamtheit'' previously introduced by von Neumann \cite{von-Neumann-1927} to analyze the statistical and thermodynamical properties of quantum systems.  
%%%%
The previous analysis was completed by proving that
\begin{equation}
    p(P) = \mathrm{Tr}(\hat{\rho} \hat{P})\,,
\end{equation}
where $\hat{P}$ is the orthogonal projection associated with $P$, in the Hilbert space $\mathcal{H}$ representing the lattice $\mathcal{P}$, and $\hat{\rho}$ is a density operator, i.e., a mixed state, or just a `quantum state' in modern jargon. 

\section{The Groupoid formalism for physical systems}

\subsection{Schwinger's selective measurements and groupoids}

One of the main difficulties with the propositional calculus of Birkhoff and von Neumann, was its explicit construction in specific examples. 
%%%
Actually, in spite of his foundational soundness, the propositional calculus developed by Birkhoff and von Neumann was seldomly used to develop the theoretical foundations of the avalanche of experimental data that were obtained at that time (contrary to Dirac's (often formal) methods that were so much easier to use). 
%%%

Almost 20 years later, J. Schwinger produced an axiomatic description of the foundations of quantum mechanics that was somehow related to the propositional calculus of Birkhoff and von Neumann even if no reference was made of the latter.
%%%
Schwinger's main idea was to look for a ``symbolic language suitable to describe atomic phenomena'' \cite{Schwinger-2000}. 
%%%
He was able to find the building blocks of such language in the notion of ``selective measurement'', which can be considered a sort of enrichment of the notion of proposition or Yes-No experiment of the propositional calculus.
%%%
In particular, denoting with $A$ a physical quantity that has the possible values $(a,a',a'',....)$, the selective measurement $M(a)$ was conceived as a specific physical apparatus that measures $A$ and selects only those physical systems compatible with the value $a$ (think, for instance, of a Stern-Gerlach device selecting only the up-beam).
%%%%%%
This notion was then extended to consider devices capable of accepting systems compatible with the value $a$ and changing them into systems compatible with  the value $a'$. 
%%%
This type of selective measurement is denoted by $M(a',a)$, and we clearly see that if $A$ has only two possible values, then $M(a,a)$ is a proposition in  von Neumann's sense.
%%%
However, and this is precisely the main contribution in Schwinger's conceptualization, instead of considering the family of selective measurements $M(a',a)$ as a static kinematical framework, Schwinger added a dynamical interpretation to it by realizing that such selective measurements may be composed in a natural way by performing them one after the other according to the composition rule
\be\label{eqn: composition of selective measurements}
M(a'', a')\,\circ\, M(a',a) \,=\, M(a'', a)\,.
\ee
%%%
Then, based on this simple rule, and a few additional natural axioms, Schwinger built  an algebra, called the algebra of selective measurements, and to put it to work he enriched it with a quantum dynamical principle that allowed him to successfully complete his analysis of quantum electrodynamics \cite{Sc51-1,Sc51-2,Sc51-3,Sc51-4,Sc51-5,Sc51-6}. 
%%%

When Schwinger introduced this theoretical framework, the mathematical structure encoded in equation \eqref{eqn: composition of selective measurements} was not known to him.
%%%
However, from a modern perspective, it is immediate to check that the symbolic composition rule given in equation \eqref{eqn: composition of selective measurements} satisfies the following axioms:
\begin{itemize}
    \item Associativity: $M(a''', a'')\circ \left( M(a'', a') \circ M(a' , a) \right) = \left( M(a''', a'') \circ M(a'' , a') \right) \circ M(a', a)$;
    \item Units: $M(a',a')\circ M(a',a) = M(a',a)$ and $M(a',a)\circ M(a,a) = M(a',a)$;
    \item Inverse: $M(a,a')\circ M(a',a) = M(a,a)$ and $M(a',a)\circ M(a,a') = M(a',a')$.
\end{itemize}
%%%
This implies that the collection of all selective measurements $\left\lbrace M(a',a) \, | \, a,\,a' \, \mbox{ outcomes of } A \right\rbrace$ form a groupoid with space of objects the measurements $M(a,a)\equiv\mathbf{1}_a$, and Schwinger's algebra of selective measurements turns out to be the groupoid algebra of this groupoid \cite{Renault-1980}.
%%%%

Starting from the observation that selective measurements are appropriately described by groupoids, in a series of recent papers \cite{C-DC-I-M-2020,C-DC-I-M-02-2020,C-I-M-2018,C-I-M-02-2019,C-I-M-03-2019,C-I-M-05-2019}, a new picture of Quantum Mechanics  has been proposed    where Schwinger's algebra of selective measurements was taken one step forward.
%%%
In this framework, a physical system is described by means of a groupoid $\mathbf{G}\rightrightarrows \Omega$, where the set $\Omega$ is referred to as the space of ``outcomes'', and the elements of the groupoid $\alpha \, : \, x \, \rightarrow \, y$, with $x,\,y \in \Omega$, $x$ being the source and $y$ the target of $\alpha$, are referred to as ``transitions''. 
%%%
In the following, we will use either the notation $\mathbf{G}\rightrightarrows \Omega$ or just $\mathbf{G}$, to denote the groupoid of a physical system.
%%%

Conceptually speaking, the family of transitions $\alpha \in \mathbf{G}$ generalize both Schwinger's notion of selective measurement previously discussed, the actual transitions used in the statement of the Ritz-Rydberg combination principle \cite{Connes-1994},   and the experimental notion of `quantum jumps' introduced in the old Quantum Mechanics. 
%%%
Furthermore, from a modern perspective, we may also say that the transitions $\alpha\,:\,x \, \rightarrow \, y$ represent the abstract notion of amplitudes as ``square roots'' of probability densities as argued in \cite{C-DC-I-M-02-2020}, that is, specific representations of the groupoid  $\mathbf{G}$ will assign rank-one operators to the transtions $\alpha$, that will represent the `square roots' of standard probabilities. 
%%%%

%%%%%%%%%%%
%%%%%%%%%%%

\subsection{The algebra of transitions and the Birkhoff-von Neumann's algebra of propositions}

In this section we will see how to construct an algebra associated with the groupoid $\mathbf{G}\rightrightarrows \Omega$ describing a physical system, and we will analyze the relation between this groupoid algebra   and  the Birkhoff-von Neumann's algebra of propositions introduced before.
%%%
For the sake of simplicity, we will consider discrete, countable groupoids.
%%%
A similar construction is availabe for more general groupoids but requires some additional care to handle  functional analytical  details \cite{Renault-1980}.
%%%%

First of all, starting from the groupoid $\mathbf{G}$, we may form the algebra $\mathbb{C}[\mathbf{G}]$ of formal linear combinations of transitions $\alpha \, \colon \, x\, \rightarrow \, y$, which is referred to as the algebra of transitions of the groupoid.
%%%
An element of this algebra will have the form  
\be\label{eqn: element of algebra of transitions}
A = \sum_{\alpha \in \mathbf{G}} A_{\alpha} \alpha, \mbox{ with }A_{\alpha}\in \mathbb{C} \, ,
\ee
(all $A_\alpha$'s zero except for a finite number of them) and the space will be equipped with a natural multiplicative law given by 
\be
    A B = \sum_{\alpha, \beta \in \mathbf{G}} A_{\alpha} B_{\beta}\,  \delta_{(\alpha, \beta)} \, \alpha \circ \beta\,,
\ee
where $\delta_{(\alpha, \beta)}$ is different from zero only if $s(\alpha) = t(\beta)$, $s$ and $t$ being the source and target map of the groupoid, respectively.
%%%
Clearly, the multiplication above is associative and, in addition, there is a natural involution operation, denoted by $\ast$ and given by 
\be
    A^{\ast} = \sum_{\alpha \in \mathbf{G}} \overline{A}_{\alpha} \alpha^{-1}\,.
\ee
The involution operation is such that $(AB)^{\ast}=B^{\ast}A^{\ast}$.
%%%
Elements of this algebra are referred to as ``virtual transitions'', and they provide the natural background for a dynamical and statistical interpretation of the theory (see \cite{C-I-M-03-2019,C-I-M-05-2019} for details).
%%%%

On the other hand, the algebra of transitions also represents a bridge between the Groupoid Picture of Quantum Mechanics and the Birkhoff-von Neumann's propositional calculus because  idempotent elements in $\mathbb{C}[\mathbf{G}]$ are related to the notion of proposition in the Birkhoff-von Neumann's lattice of propositions $\mathcal{P}$ associated with a physical system. 
%%%
Indeed, given a groupoid $\mathbf{G}$, an element $P\in\mathbb{C}[\mathbf{G}]$ such that $P^{\ast} = P$ and $P^2 = P$, will be called a proposition of $\mathbf{G}$. 
%%%
Note that, in particular, the units $\mathbf{1}_{x}$ with $x\in\Omega$ are propositions. 
%%%
There is a natural partial order relation among propositions in $\mathbf{G}$ given by
\be
    P\subseteq Q \, \Leftrightarrow \, PQ = P\,.
\ee
%%%%
Note that, if $P$ and $Q$ are propositions of $\mathbf{G}$, and $PQ = P$, then $PQ = P = P^* = Q^* P^* = QP$.
%%%
It is immediate to check that the previous relation  actually defines a partial order in the space of propositions of $\mathbf{G}$, that is, the relation $\subseteq$ thus defined, satisfies:
\begin{itemize}
    \item Reflexivity: $P^2 = P$ $\Rightarrow$ $P\subseteq P$;
    \item Transitivity: $P\subseteq Q\subseteq R$, i.e. $PQ=P$ and $QR=Q$, $\Rightarrow$ $PR=PQR=PQ=P$, i.e., $P\subseteq R$;
    \item Antisymmetry: $P\subseteq Q$ and $Q\subseteq P$, i.e. $PQ=P$ and $QP=Q$, $\Rightarrow$ $P=PQ=QP=Q$, i.e., $P=Q$. 
\end{itemize}
%%%
If the set $\Omega$ of outcomes is finite, the algebra of propositions has a unit  element $\mathbf{1}= \sum_{x\in \Omega} \mathbf{1}_x$, which may be interpreted as the proposition $E$ corresponding to ``truth''.
%%%%
Then, we can define an orthocomplement operation $P\,\mapsto\,P'$ given by:
\be
P' := \mathbf{1}- P,
\ee
It is straightforward to check that:
\begin{itemize}
    \item $\mathbf{1}' = \emptyset = 0$; \quad $0' = \mathbf{1}$;
    \item $P\subseteq Q$ $\Leftrightarrow$ $Q'\subseteq P'$ (as it follows from: $Q'P'= (\mathbf{1}-Q)(\mathbf{1}-P) = \mathbf{1}-P-Q+PQ = \mathbf{1}-Q = Q' $).
\end{itemize}
%%%%
The space of propositions of $\mathbf{G}$ would be a lattice if there exists the greatest lower bound (g.l.b.)  for any pair $(P,Q)$ of propositions, as it is the case if $\mathbf{G}$ is finite.
%%%
However, in the infinite dimensional case this is not necessarily so. 
%%%%

In the context of $C^{*}$-algebra theory, it has been recently shown by  Marchetti and Rubele \cite{M-R-2007} that there is a particular type of $C^*$-algebras, named Baire*-algebras, or $B^*$-algebras for short, for which the lattice of projection is an orthocomplemented modular lattice. In particular von Neumann algebras are $B^*$-algebras. 
%%%%
Therefore, given a groupoid $\mathbf{G}$, if we complete the algebra $\mathbb{C}[\mathbf{G}]$ of virtual transitions in such a way that it becomes a von Neumann algebra, then, its lattice of projections will satisfy the axioms of Birkhoff-von Neumann's propositional calculus.
%%%%
Notice that, in this setting, the proposition (projection) $P\cap Q$ is obtained as 
\be\label{eqn: cap for propositions}
P\cap Q\,=\,\mathrm{s-}\lim_{n\rightarrow \infty} (PQ)^n \, ,
\ee
where s-$\lim_{n\rightarrow \infty} T_n$ denotes the limit of the sequence $T_n$ in the strong topology,  and it satisfies   
\be
P\cap P' = \mathrm{s-}\lim_{n\rightarrow \infty}(P(\mathbf{1}-P))^n = 0\,.
\ee
%%%
Consequently, setting
\be \label{eqn: cup for propositions}
P\cup Q \,:= \,(P'\cap Q')',
\ee
we get an othocomplemented lattice that will be denoted as $\mathcal{P}(\mathbf{G})$. 
%%%
Note that, in addition, this lattice is automatically weakly modular because
\be
P\subseteq Q \; \Rightarrow \, PQ=QP=P,
\ee
and the sublattice generated by $P$ and $Q$ is Boolean, that is, is such that 
\begin{equation}\label{eq:distributive}
P\cup (Q\cap P) = (P\cup Q)\cap P \quad \mbox{ and }\;\;P\cap (Q\cup P) = (P\cap Q)\cup P,
\end{equation}
as can be easily checked.
%%%
%Indeed, we have
%\be
%Q\cup P= \mathbf{1}-((\mathbf{1}-Q)\cap (\mathbf{1}-P)) = \mathbf{1}-(\mathbf{1}-Q) = Q
%\ee
%and $P\cap (Q\cup P)=P\cap Q=P$ as well as $(P\cup Q)\cap P = Q\cap P= P$. Similarly, $P\cup (Q\cap P)= P\cup P = P$ and $(P\cup Q)\cap P= Q\cap P = P$.
%%%%

In the particular instance that the groupoid $\mathbf{G}$ is countable, as we shall always assume from now on, we may easily complete $\mathbb{C}[\mathbf{G}]$ to a von Neumann algebra $\valg$ associated with $\mathbf{G}$ as follows.
%%%
Consider the space $\mathcal{L}^2(\mathbf{G})$ of square integrable functions on $\mathbf{G}$ with respect to the counting measure. 
%%%
There is an algebra homomorphism  $\lambda$ between the algebra $\mathbb{C}[\mathbf{G}]$ of transitions of $\mathbf{G}$  and the algebra $\mathcal{B}(\mathcal{L}^2(\mathbf{G}))$ given by
\be
(\lambda(A)\Psi)(\beta) \,:=\,\sum_{\alpha\in\mathbf{G}}\, A_{\alpha}\,\delta_{(\alpha^{-1}, \beta)}\Psi (\alpha^{-1}\circ \beta)\, ,
\ee
that is, $(\lambda(A)\Psi)(\beta) =\,\sum_{\alpha\in\mathbf{G}}\, A_{\alpha}\,\Psi (\alpha^{-1}\circ \beta)$,
provided that $t(\alpha)=t(\beta)$, and zero otherwise.
%%%
The map $\lambda$ is called the left-regular representation of $\mathbb{C}[\mathbf{G}]$ (the notions of representation of a groupoid and of its algebra are connected with the notions of representation of a category and of its associated algebra as explained in \cite{I-R-2019}).
%%%
The representation $\lambda$ is faithful, and thus, in the following, we will often identify $\mathbb{C}[\mathbf{G}]$ with its image through $\lambda$.
%%%
This choice will simplify the notation.
%%%%

We define the  algebra $\valg$  as 
\be
\valg\,:=\,\left(\lambda(\mathbb{C}[\mathbf{G}])\right)'',
\ee
that is, as the double commutant of $\lambda(\mathbb{C}[\mathbf{G}])$ inside $\mathcal{B}(\mathcal{L}^2(\mathbf{G}))$.
%%%
Because of von Neumann's theorem, the algebra $\mathcal{V}(\mathbf{G}) \subset \mathcal{B}(\mathcal{L}^2(\mathbf{G}))$ is a weakly closed and strongly-closed subalgebra of the algebra of bounded operators on the separable Hilbert space $\mathcal{L}^2(\mathbf{G})$, and thus it is a von Neumann algebra.
%%%
The identity operator in $\mathcal{V}(\mathbf{G})$ is denoted as $\mathbf{1}$, and it is not hard to see that $\mathbf{1}= \sum_{x\in \Omega} 1_x $. 
%%%
We refer to $\valg$ as the von Neumann algebra of the groupoid  $\mathbf{G}$.

Thus, we may conclude this section by saying that the set of propositions (projections) of the von Neumann algebra $ \mathcal{V}(\mathbf{G})$ of virtual transitions of a physical system determined by a countable groupoid $\mathbf{G}$ is an orthocomplemented, atomic, weakly-modular lattice of propositions in the sense of Birkhoff-von Neumann's propositional calculus.

\subsection{States in the groupoid picture}

As it was discussed in Sect. \ref{sec:propositions},   from the point of view of the algebra of propositions, a state of a physical system is a non-negative, normalized, $\sigma$-additive function $p$ on the lattice of propositions $\mathcal{P}$ of a physical system. 
%%%
Therefore, in the groupoid picture outlined above, a state $p$ will be a normalized non-negative real function on the lattice of propositions $\mathcal{P}(\mathbf{G})$ of the von Neumann algebra $\valg$ of the groupoid $\mathbf{G}\rightrightarrows \Omega$.
%%%
In other words, $\forall P \in \mathcal{P}(\mathbf{G})$, $p(\mathbf{1})=1$ and $p(P)\geq 0$. 
%%%
Moreover, since the lattice of propositions (projections) $\mathcal{P}(\mathbf{G})$ generates the total algebra $\valg$ (essentially because of the spectral theorem), we have that any real element $A \in \valg$  can be written as 
\be
A=\sum_{\lambda} \lambda P_{\lambda},
\ee
where $\left\lbrace P_{\lambda} \right\rbrace$ is the spectral resolution of $A$. Then $p(A) = \sum_{\lambda}\lambda p(P_{\lambda})$. If $A$ is also positive, i.e., $A=B^*B$, we will have $p(B^*B)= \sum_{\lambda}|\lambda|^2p(P_{\lambda}) \geq 0$. In other words, states in the sense of von Neumann are just states in the $C^*$-algebra $\valg$, i.e., normalized positive functionals. This justifies the notion of state in the groupoid picture introduced in \cite{C-I-M-05-2019} as normalized, positive functionals defined on the $C^*$-algebra of the groupoid $\mathbf{G}$.
%%%%
The case in which $A$ has a continuous spectrum may be dealt with in a similar way, essentially ``replacing sums with integrals''.
%%%%

In the case of a countable groupoid, every state $\rho$ determines a  function $\varphi_{\rho}\, :\, \mathbf{G} \, \rightarrow \, \mathbb{C}$ which is positive-definite and is given by 
\be
\varphi_{\rho}(\alpha) := \rho(\alpha)\,, \quad \forall \alpha \in \mathbf{G}\,.
\ee
%%%
Note that we have 
\be
\varphi_{\rho}(1_x)=\rho(1_x)= \rho(1_x^*1_x)\geq 0 ,
\ee
and 
\be
\sum_{x\in \Omega} \varphi_{\rho}(1_x) = \rho(\mathbf{1})=1 .
\ee
%%%
Hence, the non-negative real numbers, 
\be
p_x = \varphi_{\rho}(1_x) \quad \mbox{ with } x\in \Omega\,,
\ee
define a classical probability distribution on the space of outcomes  $\Omega$ of the system. 
%%%

As it was shown in \cite{C-I-M-05-2019},  any state $\rho$ on $\valg$, through its associated positive-definite function $\varphi_{\rho}$, defines a decoherence functional on the $\sigma$-algebra $\Sigma (\mathbf{G})$ of parts of $\mathbf{G}$ by means of
\be
D(\mathbf{A},\mathbf{B}) = \sum_{\substack{\alpha\in \mathbf{A}, \, \beta\in \mathbf{B}\\t(\alpha)=t(\beta)}} \varphi_{\rho}(\alpha^{-1}\circ \beta)\,, \quad \mathbf{A},\,\mathbf{B}\in \Sigma(\mathbf{G})\,.
\ee
%%%
The decoherence functional $D$, in turn, defines a quantum measure, or a grade-2 measure $\mu$ in Sorkin's conceptualization of the statistical interpretation of Quantum Mechanics \cite{Sorkin-1994}, as
\be
\mu (\mathbf{A})= D(\mathbf{A}, \mathbf{A})\,. 
\ee
%%%
The quantum measure $\mu$ nicely captures both the statistical interpretation of experimental observations (it can be interpreted as a statistical frequence) as well as interference phenomena, i.e., it is not additive in general (see \cite{C-I-M-05-2019} and references therein for a detailed discussion of these subtle aspects).
%%%
Therefore, in general we can define the interference function 
\be
I_2(\mathbf{A},\mathbf{B}):= \mu(\mathbf{A}\sqcup \mathbf{B}) - \mu(\mathbf{A})-\mu(\mathbf{B}) \neq 0\,,
\ee
where $\mathbf{A},\, \mathbf{B} \in \Sigma(\mathbf{G})$ are disjoint subsert of $\mathbf{G}$. 
%%%
However, a grade-2 measure satisfies the following identity 
\be
I_3(\mathbf{A},\mathbf{B},\mathbf{C}) := \mu(\mathbf{A}\sqcup \mathbf{B}\sqcup \mathbf{C}) - \mu(\mathbf{A}\sqcup \mathbf{B}) - \mu(\mathbf{A}\sqcup \mathbf{C}) - \mu(\mathbf{B}\sqcup \mathbf{C}) + \mu(\mathbf{A})+\mu(\mathbf{B})+\mu(\mathbf{C}) = 0\,,
\ee
with $\mathbf{A},\,\mathbf{B},\,\mathbf{C}\in \Sigma(\mathbf{G})$, three  pairwise disjoint subsets of $\mathbf{G}$. 
%%%%

Given two outcomes $x,y \in \Omega$ and a state $\rho$ of the system, we can consider the complex number  
\be
\varphi_{y,x} = \sum_{\substack{\alpha:x\rightarrow y}} \varphi_{\rho}(\alpha)\,, 
\ee
which could be interpreted as the probability amplitude associated with the state $\rho$ and the events $x$ and $y$
%%%
Indeed, we have
\be
\begin{split}
|\varphi_{y,x}|^2 = \overline{\varphi}_{y,x}\varphi_{y,x} & = \sum_{\substack{\alpha:x\rightarrow y}}\overline{\varphi}_{\rho}(\alpha) \sum_{\substack{\beta:x\rightarrow y}}\varphi_{\rho}(\beta)= \\
& = \sum_{\substack{\alpha:x\rightarrow y \\ \beta:x\rightarrow y}}\varphi_{\rho}(\alpha^{-1})\varphi_{\rho}(\beta)\,.
\end{split}
\ee
%%%
Then, if the state $\rho$ is factorizable in the sense of  \cite{C-I-M-05-2019}, i.e., if its associated function $\varphi_{\rho}$ satisfies
\be\label{eqn: factorizable state}
\varphi_{\rho}(\alpha\circ \beta) = \varphi_{\rho}(\alpha)\varphi_{\rho}(\beta),
\ee
we obtain that
\begin{equation}
|\varphi_{y,x}|^2 = \sum_{\substack{\alpha:x\rightarrow y\\ \beta:x\rightarrow y}}\varphi_{\rho}(\alpha^{-1}\circ \beta)\,,
\end{equation}
which is the celebrated rule of ``sum-over-histories'' composition of the probability amplitudes discovered by R. Feynman. 
%%%

\section{Composition of classical and quantum systems}

\subsection{Classical systems}\label{sec:classical}

To analyze Schr\"{o}dinger's cat states, we must first consider with care the notion of classical system.
%%%
In Birkhoff-von Neumann's description of physical systems by means of propositional calculus, classical systems correspond to Boolean lattices.
%%%
More precisely, two propositions  $P$ and $Q$  are said to be compatible if the sublattice generated by them is Boolean \cite{Jauch-1952}, i.e., it satisfies the distributive law in Eq. \ref{eq:distributive}, and a system is classical  if all the propositions are compatible among themselves. 
%%%%
On the other hand, in the groupoid picture introduced above, we   say that a given physical system is classical if the algebra $\valg$ is Abelian (commutative).
%%%%
%Abelian von Neumann algebras associated with a groupoid $\mathbf{G}$ have the form $\mathcal{L}^{\infty}(\hat{\mathbf{G}})$, where $\hat{\mathbf{G}}$ is the space of representations of $\mathbf{G}$, and propositions are just characteristic functions on $\hat{\mathbf{G}}$.
%%%%
%Indeed, if $P^2=P$,  the set $P^{-1}(1)=\Delta$ is such that $P=\chi_{\Delta}$, $\Delta \subset \hat{\mathbf{G}}$ (notice that $P$ can only take either the value 1 or the value 0). Moreover, $\chi_{\Delta}\cap\chi_{\Delta'} = \chi_{\Delta}\chi_{\Delta'} = \chi_{\Delta\cap \Delta'}$ and $\chi_{\Delta}\cup \chi_{\Delta'}= \chi_{\Delta\cup \Delta'}$, and the lattice of propositions is obviously Boolean. 
%%%%

We will now prove that the notion of classical system  from the groupoid point of view  corresponds to the notion of classical system  from Birkhoff and von Neumann's quantum logic point of view and {\itshape viceversa}.
%%% 

First, we will show that if the groupoid algebra $\valg$ is Abelian, then the lattice of propositions of $\mathbf{G}$ is Boolean.

\begin{proposition}\label{prop: countable groupoid with abelian VA has Boolean lattice of prop}
Let $\mathbf{G}\rightrightarrows \Omega$ be a countable groupoid with Abelian von Neumann algebra $\valg$, then its lattice of propositions $\mathcal{P}(\mathbf{G})$ is Boolean.
\end{proposition}

\begin{proof} First, notice that, if $\valg$ is Abelian, for any pair of propositions $P$ and $Q$ we have $P\cap Q = PQ = QP$.
%%%%
Indeed, it clearly holds $(PQ)^n = P^n Q^n = PQ$, and thus, recalling Eq. \eqref{eqn: cap for propositions}, it is $P\cap Q = \mathrm{s-}\lim_n (PQ)^n = PQ$. 
%%%
Then, it is
\be
(P\cup Q)' = P'\cap Q' = (\mathbf{1} - P)  (\mathbf{1} - Q) =  \mathbf{1} - P - Q + PQ,
\ee
which means
\be
P\cup Q = P + Q - PQ.
\ee
%%%
Now, we may easily check the distributive property  of the lattice $\mathcal{P}(\mathbf{G})$ (see Eq. \eqref{eqn: Boolean lattice}) by computing
\begin{eqnarray*}
P\cup (Q \cap R) &=& P \cap QR = P + QR - PQR  \\ &=&  P+  (PR - PR) + (QP - PQ) + QR - PQR + (-PQR + PQR ) \\ &=&   (P+ Q - PQ) (P + R - PR)  =   (P\cup Q) (P \cup R)   =  (P\cup Q) \cap (P \cup R) \, ,
\end{eqnarray*}
and 
\begin{eqnarray*}
P\cap (Q \cup R) &=& P (Q\cup R) = P(Q + R - QR) \\ &=&  PQ +PR - PQR =   (P Q)\cup  (P R) =  (P\cap Q) \cup (P \cap R)\, .
\end{eqnarray*}
\end{proof}

The converse is also true and, in addition, we gain additional information on the structure of the von Neumann algebra of a classical system.
%%%
%A countable groupoid $\mathbf{G}\rightrightarrows \Omega$ is amenable if the carrier vector spaces of all its irreducible representations are subspaces of the left regular representation $\lambda$ of $\mathbf{G}\rightrightarrows \Omega$ on $\mathcal{L}^{2}(\mathbf{G})$. 
%%%%

\begin{theorem}\label{thm: equivalence between classical systems in the groupoid formalism and the propositional calculus}
Let $\mathbf{G}\rightrightarrows \Omega$ be a discrete countable groupoid.
%%%
The lattice of propositions $\mathcal{P}(\mathbf{G})$ of the groupoid is Boolean if and only if the groupoid $\mathbf{G}$ is totally disconnected, its isotropy groups are Abelian, and its von Neumann algebra $\valg$ is Abelian.
\end{theorem} 

\begin{proof}
In Prop. \ref{prop: countable groupoid with abelian VA has Boolean lattice of prop} we have already seen that if $\valg$ is Abelian, then $\mathcal{P}(\mathbf{G})$ is Boolean. Now, let us prove first that if $\mathcal{P}(\mathbf{G})$  is Boolean, then $\mathbf{G}$ is totally disconnected.
%%%
We will do it by showing that, given $x\neq y \in \Omega$, there is no transition $\alpha$ between $x$ and $y$ (i.e., there are no ``quantum jumps'' from $x$ to $y$). 
%%%%
Suppose that such a transition $\alpha\,:\,x\rightarrow y$, with $x\neq y$, exists.
%%%
Then, let us define the proposition 
\be
P_{\alpha}\,:= \, \frac{1}{2}\left( 1_x+1_y+\alpha+\alpha^{-1} \right).
\ee
%%%
%Clearly, $P_{\alpha}^*=P_{\alpha}$ and $P^2_{\alpha}=P_{\alpha}$. 
%%%
Now, notice that we have
\be
1_xP_{\alpha} = \frac{1}{2}\left( 1_x + \alpha^{-1} \right),
\ee
from which it follows that (see equation \eqref{eqn: cap for propositions})
\begin{equation}
1_x\cap P_{\alpha}= \mathrm{s-}\lim_{\substack{n\rightarrow \infty}} (1_x P_{\alpha})^n = \mathrm{s-}\lim_{\substack{n\rightarrow \infty}} \frac{1}{2^n}\left( 1_x + \alpha^{-1} \right) = 0\,.
\end{equation}
%%%
On the other hand, we have (see Eq. \eqref{eqn: cup for propositions}):
\be
(1_x\cup 1_y)' = 1_x' \cap 1_y' = (\mathbf{1} - 1_x ) \cap  (\mathbf{1} - 1_y )  = s-\lim_n  ((\mathbf{1} - 1_x )   (\mathbf{1} - 1_y ))^n  = \mathbf{1} -  (1_x+1_y)\, ,
\ee
and, hence,
$$
1_x\cup 1_y = 1_x + 1_y \, .
$$
Moreover,
\be
P_{\alpha}(1_x+1_y) = 2 P_{\alpha}.
\ee
%%%
Therefore, we conclude that
\begin{equation*}
P_{\alpha} \cap (1_x\cup 1_y) = P_{\alpha}\cap (1_x+1_y)= P_{\alpha} ,
\end{equation*}
and, on the other hand, that:
\begin{equation*}
(P_{\alpha}\cap 1_x)\cup(P_{\alpha}\cap 1_y)=0\cup 0= 0\,,
\end{equation*}
and these two results are obviously different.
%%%
This means that, if there is a ``quantum jump'', we can construct two propositions which are not compatible, and this is a contradiction with the hypothesis of the theorem.
%%%
Consequently, $\mathbf{G}\rightrightarrows \Omega$ is totally disconnected and 
\be
\mathbf{G}= \sqcup_{\substack{x\in \Omega}} \mathbf{G}_x \, ,
\ee
with $\mathbf{G}_x$ being the isotropy group at $x$. 
%%%%
According to \cite[prop. 10.12]{I-R-2019}, the algebra $\mathbb{C}[\mathbf{G}]$ of transitions of the totally disconnected groupoid $\mathbf{G}$ can be written as
\be
\mathbb{C}[\mathbf{G}]\,=\,\bigoplus_{x\in\Omega}\,\mathbb{C}[\mathbf{G}_{x}]\,,
\ee
where $\mathbb{C}[\mathbf{G}_{x}]$ is the algebra of transitions of the isotropy group $\mathbf{G}_{x}$ seen as a groupoid with only one object.
%%%%
Then, the von Neumann algebra of $\mathbf{G}$ will have the form
\be
\mathcal{V}(\mathbf{G})\,=\,\bigoplus_{x\in\Omega}\,\mathcal{V}(\mathbf{G}_{x})\,.
\ee
%%%
By hypothesis, $\mathbf{G}_x$ is a countable discrete group, and  the lattice of propositions $\mathcal{P}(\mathbf{G}_x)$ is a sublattice of the algebra of propositions $\mathcal{P}(\mathbf{G})$, hence, it is a Boolean sublattice of the Boolean lattice $\mathcal{P}(\mathbf{G})$.   The von Neumann algebra $\mathcal{V}(\mathbf{G}_{x})$ is generated by its projectors, that by hypothesis are compatible propositions, hence they commute (see \cite[Sect. 5.8]{Jauch-1952}) and  $\mathcal{V}(\mathbf{G}_{x})$ is Abelian.
%%%
Hence, all isotropy groups $\mathbf{G}_x$ must be Abelian and  we have
\begin{equation}
\mathcal{V}(\mathbf{G}) = \bigoplus_{\substack{x\in \Omega}}\mathcal{V}(\mathbf{G}_x) \cong \bigoplus_{\substack{x\in \Omega}} \mathcal{L}^{\infty}(\widehat{\mathbf{G}}_x) =\mathcal{L}^{\infty} \left(\sqcup_{\substack{x\in \Omega}} \widehat{\mathbf{G}}_x \right)\,,
\end{equation}
with $\widehat{\mathbf{G}}_x$ being the Pontryagin's dual group (or group of characters) of ${\mathbf{G}}_x$.  
\end{proof}

From this, we conclude that classical systems (at least those described by countable groupoids),  both in the groupoid picture and in the   quantum logic approach of Birkhoff and von Neumann, correspond to the same notion and are described by Abelian von Neumann algebras that can always be realized as $\mathcal{L}^{\infty}(X)$ with $X = \sqcup_{\substack{x\in \Omega}} \hat{\mathbf{G}}_x$.
%%%
The states of these algebras  are non-negative, normalized integrable functions on $X $, i.e., functions $\rho\, :\, X \, \rightarrow \,\mathbb{R}_+$, such that 
\be
\int_X \rho(x)d\mu(x) = 1\,.
\ee

%%%%%%%%%%%%%
%%%%%%%%%%%%%

\subsection{Composition}

In order to meaningfully discuss the situation depicted by Schr\"{o}dinger, we must consider the composition of two systems, in particular, of a classical system with a quantum one.
%%%%
We will carry on the analysis of composite systems in the groupoid framework discussed above.
%%%%

The composition of systems we will consider here is the simplest one (or the na\"{i}ve one), which intuitively corresponds to the idea ``to put the systems side by side on the laboratory table'' \cite{L-Y-1999,L-Y-2000}, i.e., all possible outcomes of both systems can be determined simultaneously, and all possible transitions of both systems can actually happen. 
%%%%
As thoroughly discussed in \cite{C-DC-I-M-2020}, in the language of groupoids, this  notion of composition corresponds to the direct product of groupoids.
%%%
Specifically, given two systems A and B described by two groupoids $\mathbf{G}_A\rightrightarrows \Omega_A$ and $\mathbf{G}_B \rightrightarrows \Omega_B$ respectively, their direct composition, denoted by $A\times B$,   corresponds to the system described by the groupoid $\mathbf{G}_A\times \mathbf{G}_B \rightrightarrows \Omega_A \times \Omega_B$, which is the direct product of the groupoids $\mathbf{G}_A$ and $\mathbf{G}_B$. 
%%%
The units of $\mathbf{G}_A\times \mathbf{G}_B$ are given by the pairs $(1_{x_A},1_{x_B})$, and the inverse of the transition $(\alpha_A, \alpha_B)$ is $(\alpha_A,\alpha_B)^{-1} = (\alpha_A^{-1},\alpha_B^{-1})$.
%%%%

Our aim is now to prove that the algebra of transitions of the composite groupoid  is the (algebraic) tensor product of the algebras of transitions of the composing groupoids.
%%%

\begin{proposition}
Given the  countable groupoids $\mathbf{G}_A\rightrightarrows \Omega_A$, $\mathbf{G}_B \rightrightarrows \Omega_B$, and $\mathbf{G}_A\times \mathbf{G}_B \rightrightarrows \Omega_A \times \Omega_B$, it holds
\be
\mathbb{C}[\mathbf{G}_A\times \mathbf{G}_B]\, \cong\, \mathbb{C}[\mathbf{G}_A]\,\otimes\,\mathbb{C}[\mathbf{G}_B]\,.
\ee
\end{proposition}
\begin{proof}
Consider the algebraic tensor product $ \mathbb{C}[\mathbf{G}_A]\,\otimes\,\mathbb{C}[\mathbf{G}_B]$ of the algebras of transitions of the groupoids, and  consider the   map $i_{AB}\colon  \mathbb{C}[\mathbf{G}_A \times  \mathbf{G}_B] \ra \mathbb{C}[\mathbf{G}_A]\,\otimes\,\mathbb{C}[\mathbf{G}_B]$ obtained by extending by linearity the map 
\be
(\alpha_{A},\beta_{B})\,\mapsto \,i_{AB} (\alpha_{A},\beta_{B})\,:=\,\alpha_{A}\otimes\beta_{B}.
\ee
%%%
Specifically, denoting by $(\gamma_{A},\gamma_{B})$ an element in the direct product $\mathbf{G}_{A}\times\mathbf{G}_{B} \rightrightarrows \Omega_A \times \Omega_B$, every  element in the algebra of transition can be written as (see Eq. \eqref{eqn: element of algebra of transitions})
\be
a=\sum_{(\gamma_{A},\gamma_{B})\in\mathbf{G}_A\times\mathbf{G}_B}\,a_{(\gamma_{A},\gamma_{B})}\,(\gamma_{A},\gamma_{B}),
\ee
and the map $i_{AB}$ reads
\be
i_{AB}(a)\,=\,\sum_{(\gamma_{A},\gamma_{B})\in\mathbf{G}_A\times\mathbf{G}_B}\,a_{(\gamma_{A},\gamma_{B})}\,\gamma_{A}\,\otimes\,\gamma_{B}.
\ee
%%%%
It immediately follows that $i_{AB}$ is injective, and we will now see that it is actually surjective.
%%%%
At this purpose, we consider a generic element 
\be
c=\sum_{j=1}^{N}\,c_{j}\,a_{A}^{j}\,\otimes\,b_{B}^{j}
\ee
in the algebraic tensor product  $\mathbb{C}[\mathbf{G}_A]\,\otimes\,\mathbb{C}[\mathbf{G}_B]$.
%%%
Note that $N$ is finite and depends on the element $c$.
%%%
Now, since $a_{A}^{j}\in\mathbb{C}[\mathbf{G}_A]$ and $b_{B}^{j}\in\mathbb{C}[\mathbf{G}_B]$, we have
\be
\begin{split}
a_{A}^{j}&\,=\,\sum_{\gamma_{A}\in\mathbf{G}_{A}}\,a_{\gamma_{A}}^{j}\,\gamma_{A}  \\
b_{B}^{j}&\,=\,\sum_{\gamma_{B}\in\mathbf{G}_{B}}\,b_{\gamma_{B}}^{j}\,\gamma_{B}  ,
\end{split}
\ee
and thus
\be
c=\sum_{j=1}^{N}\,\sum_{(\gamma_{A},\gamma_{B})\in\mathbf{G}_{A}\times\mathbf{G}_{B}} \,c_{j}\,a_{\gamma_{A}}^{j}\,b_{\gamma_{B}}^{j}\,\gamma_{A}  \,\otimes\,\gamma_{B}   \,.
\ee
%%%
Setting
\be
c_{(\gamma_{A},\gamma_{B})}\,:=\,\sum_{j=1}^{N}\,c_{j}\,a_{\gamma_{A}}^{j}\,b_{\gamma_{B}}^{j},
\ee
it is clear that the element
\be
\tilde{c}\,=\,\sum_{(\gamma_{A},\gamma_{B})\in\mathbf{G}_{A}\times\mathbf{G}_{B}}\,c_{(\gamma_{A},\gamma_{B})}\,(\gamma_{A},\gamma_{B})
\ee
in $\mathbb{C}[\mathbf{G}_{A}\times\mathbf{G}_{B}]$ is such that
\be
i_{AB}(\tilde{c})\,=\,c,
\ee
and thus $i_{AB}$ is surjective.
%%%
From this, we conclude that
\be
\mathbb{C}[\mathbf{G}_A\times \mathbf{G}_B]\, \cong\, \mathbb{C}[\mathbf{G}_A]\,\otimes\,\mathbb{C}[\mathbf{G}_B]\,,
\ee
as desired.
\end{proof}

%%%
Note that the unit of $\mathbb{C}[\mathbf{G}_A\times \mathbf{G}_B]$ is given by
\be
\mathbf{1}_{A\times B} = \sum_{\substack{x_A\in \Omega_A\\ x_B \in \Omega_B}}\left( 1_{x_A}, 1_{x_B} \right) = \left( \mathbf{1}_A, \mathbf{1}_B \right) \,.
\ee
%%%%
If we use the isomorphism $\mathbb{C}[\mathbf{G}_A\times \mathbf{G}_B]\, \cong\, \mathbb{C}[\mathbf{G}_A]\,\otimes\,\mathbb{C}[\mathbf{G}_B]$ introduced above,  we can write
\be
\mathbf{1}_{A\times B} = \mathbf{1}_A\otimes \mathbf{1}_B\,.
\ee
%%%%
Note that the left-regular representation of $\mathbb{C}[\mathbf{G}_A\times \mathbf{G}_B]$ is supported on the Hilbert space
\begin{equation}\label{eq:tensor}
\mathcal{L}^{2}(\mathbf{G}_{A}\times\mathbf{G}_{B})\,\cong\,\mathcal{L}^{2}(\mathbf{G}_{A})\, \otimes\,\mathcal{L}^{2}(\mathbf{G}_{A}),
\end{equation}
and thus it is the tensor product of the left-regular representations of $\mathbb{C}[\mathbf{G}_{A}]$ and $\mathbb{C}[\mathbf{G}_{B}]$ respectively.
%%%%

%%%

Now, we will specialize to the composition of a classical system with a quantum system in order to analyse Schr\"{o}dinger's {\itshape Gedankenexperiment} in the next section.
%%%
Hence, if A denotes a classical system with totally disconnected Abelian groupoid 
\be
\mathbf{G}_A = \bigsqcup_{\substack{x_A\in \Omega_A}}\mathbf{G}_{x_A},
\ee 
where $\mathbf{G}_{x_A}$ is Abelian  $\forall x_A \in \Omega_A$, and B is a quantum system with groupoid $\mathbf{G}_B\rightrightarrows \Omega_B$ (the ``quantumness'' of B is encoded in the assumption that $\mathbb{C}[\mathbf{G}_{B}]$ is not Abelian),  the direct composition of both will be the groupoid $\mathbf{G}_A\times\mathbf{G}_B\rightrightarrows \Omega_A\times\Omega_B$.
%%%
The outcomes of the composite system $A\times B$ will consist  in all possible pairs of outcomes $(x_A,x_B)\in \Omega_A\times \Omega_B$, as in the general situation. 
%%%
On the other hand, the transitions of the composition will   have the form $(\gamma_{x_A}, \alpha_B)$ with $\gamma_{x_A}\in \mathbf{G}_{x_A}$ and $\alpha_B\,:\, x_B\,\rightarrow \, y_B$. 
%%%
Assuming that $\mathbf{G}_B$ is connected, the orbits of the groupoid $\mathbf{G}_A\times \mathbf{G}_B$ have the form 
\be
\mathcal{O}_{x_A} = \left\lbrace (x_A,x_B) | x_A\: \mathrm{fixed},\, x_B\in \Omega_B \right\rbrace = \left\lbrace x_A \right\rbrace\times \Omega_B,
\ee
from which we conclude that  the groupoid $\mathbf{G}_A\times\mathbf{G}_B$ is not connected and the space of outcomes decomposes as the disjoint union of the family of orbits $\mathcal{O}_{x_A}$ according to 
\be
\Omega_A\times \Omega_B = \bigsqcup_{\substack{x_A\in\Omega_A}}\mathcal{O}_{x_A}.
\ee
%%%%

As before, the algebra of virtual transitions of the composite system $A\times B$ has the form  
\be\label{eqn: tensor product of systems}
\mathbb{C}[\mathbf{G}_A\times \mathbf{G}_B]\, = \mathbb{C}[\mathbf{G}_A]\,\otimes\,\mathbb{C}[\mathbf{G}_B]\,,
\ee
while, to understand the structure of the von Neumann algebra $\valgab$, we need a preliminary lemma.

\begin{lemma}\label{lem:tensor}
Let $\hh_{A}$ and $\hh_{B}$ be two complex, separable Hilbert spaces, and let $A$ and $B$ be two unital $*$-subalgebras of $\mathcal{B}(\mathcal{H}_{A})$ and $\mathcal{B}(\hh_{B})$, respectively.
%%%
Let $\mathcal{V}(A)\subset\mathcal{B}(\hh_{A})$ and $\mathcal{V}(B)\subset\mathcal{B}(\hh_{B})$ denote the von Neumann algebras generated by $A$ and $B$ respectively (i.e., their double commutants or weak closures).
%%%
Then, it holds
\be
\mathcal{V}(A) \widehat{\otimes} \mathcal{V}(B) = \mathcal{V}(A\otimes B),
\ee
where $\widehat{\otimes}$ denote the tensor product of von-Neumann algebras.
\end{lemma}

\begin{proof}
We will consider first the algebraic tensor product $A \otimes B$ of  $A$ and $B$.
%%%
The resulting algebra is supported on the Hilbert space $\hh=\hh_{A}\otimes\hh_{B}$.
%%%
We denote by $\mathcal{V}(A\otimes B) = (A\otimes B)''$ the von Neumann algebra generated by it. 
%%%
Clearly,  $\mathcal{V}(A\otimes B)$ is supported on $\hh$.
%%%
The von Neumann algebra $ \mathcal{V}(A) \widehat{\otimes} \mathcal{V}(B)$ is itself also supported on $\hh$, essentially by the definition of the tensor product of von Neumann algebras \cite{Takesaki-2002}.
%%%

Now, consider the canonical inclusion $A\otimes B \subset \mathcal{V}(A) \widehat{\otimes} \mathcal{V}(B)$.
%%%
The closure in the weak topology will induce a continuous inclusion
$$
\mathcal{V}(A \otimes B) \subset \mathcal{V}(A) \widehat{\otimes} \mathcal{V}(B) \,.
$$
On the other hand, there are natural inclusions: $A \hookrightarrow A \otimes B$, $a \hookrightarrow a \otimes 1_B$, and $B \hookrightarrow A\otimes B$, $b \hookrightarrow 1_A \otimes b$, that induce inclusions $\mathcal{V}(A) \subset \mathcal{V}(A\otimes B)$, $\mathcal{V}(B) \subset \mathcal{V}(A\otimes B)$ respectively.  Consequently, $\mathcal{V}(A) \otimes \mathcal{V}(B)\subset \mathcal{V}(A\otimes B)$.  
%%%
Finally, considering the closure with respect to the weak operator topology, we get:
$$
\mathcal{V}(A) \widehat{\otimes} \mathcal{V}(B)\subset \mathcal{V}(A\otimes B) \, .
$$
\end{proof}

Now, consider the separable Hilbert spaces $\hh_{A}=\mathcal{L}^{2}(\mathbf{G}_{A})$, $\hh_{B}=\mathcal{L}^{2}(\mathbf{G}_{B})$, and $\mathcal{H}=\hh_{A}\otimes\hh_{B}$, and the unital *-subalgebras $\mathbb{C}[\mathbf{G}_{A}] $ and $ \mathbb{C}[\mathbf{G}_{B}]$ of $\mathcal{B}(\mathcal{H}_{A})$ and $\mathcal{B}(\hh_{B})$, respectively.
%%%%
Recalling Eq. \eqref{eqn: tensor product of systems}, and applying Lemma \ref{lem:tensor}, we obtain
\be
\mathcal{V}(\mathbf{G}_{A}\,\times\,\mathbf{G}_{B})\,=\,\valga\,\widehat{\otimes}\,\valgb\,.
\ee
%%%%
Moreover,  the algebra $\mathbb{C}[\mathbf{G}_{A}]$ of the classical system reads (recall the discussion in Sect. \ref{sec:classical} and \cite[prop. 10.12]{I-R-2019})
\be
\mathbb{C}[\mathbf{G}_{A}]\,=\,\bigoplus_{x_{A}\in\Omega_{A}}\,\mathbb{C}[\mathbf{G}_{x_{A}}],
\ee
and thus
\be
\valga\,=\,\bigoplus_{x_{A}\in\Omega_{A}}\,\valgxa\,.
\ee
%%%
Exploiting the distributive property of the tensor product over the direct sum we obtain 
\be\label{eqn: decomposition of the algebra of a classical and quantum system 2}
\valgab\,=\,\bigoplus_{x_{A}\in\Omega_{A}}\,\left(\valgxa\,\widehat{\otimes}\,\valgb\right)\,.
\ee
%%%

%Indeed, taking the closure with respect to the weak operator topology on both sides of  Eq. \eqref{eqn: decomposition of the algebra of a classical and quantum system}, we get
%$$
%\mathcal{V}(\mathbf{G}_A\times \mathbf{G}_B)\, = \,\bigoplus_{x_{A}\in\Omega_{A}}\,\mathcal{V}(\mathbb{C}[\mathbf{G}_{x_{A}}]\,\otimes\,\mathbb{C}[\mathbf{G}_B]) \, ,
%$$
%and then, using Lemma \ref{lem:tensor}, we get Eq. \eqref{eqn: decomposition of the algebra of a classical and quantum system 2}. 
%\be\label{eqn: decomposition of the algebra of a classical and quantum system}
%\mathbb{C}[\mathbf{G}_A\times \mathbf{G}_B]\, = \,\bigoplus_{x_{A}\in\Omega_{A}}\,\left(\mathbb{C}[\mathbf{G}_{x_{A}}]\,\otimes\,\mathbb{C}[\mathbf{G}_B]\right) \, .
%\ee
%%%

%
%
%Thus, the possible virtual transitions of the composition are given by direct sums of tensor products of the form $a_{x_{A}}\otimes b$, with $a_{x_{A}}\in \mathbb{C}[\mathbf{G}_{x_A}]$ and $b\in\mathbb{C}[\mathbf{G}_{B}]$. 
%%%%%%
%Note that every $A\in \mathbb{C}[\mathbf{G}_A\times \mathbf{G}_B]$ can be written as 
%\be\label{eqn: decomposition of elements the algebra of a classical and quantum system}
%A\,=\,\sum_{x_{A}\in\Omega_{A}}\,\sum_{j=1}^{N}\,\left(A_{x_{A}}^{j}\,\otimes\,A_{B}^{j}\right) 
%\ee
%where  $A_{x_{A}}^{j} \in \mathbb{C}[\mathbf{G}_A]$ and $A_{B}^{j} \in \mathbb{C}[\mathbf{G}_B]$.
%%%%

%%%%

The decomposition of the algebra given in equation \eqref{eqn: decomposition of the algebra of a classical and quantum system 2} will play a key role in the analysis of Schr\"{o}dinger's {\itshape Gedankenexperiment}, and it is important to point out that the direct product structure of the algebra of the composite system is manifest only in the groupoid formalism presented.
%%%
Indeed, in the algebraic approach, one starts with an Abelian algebra for the classical system and with a non-Abelian algebra for the quantum system, and then builds the algebra of the composite system as the tensor product of the algebra.
%%%
The result is an algebra in a form similar to that in equation \eqref{eqn: tensor product of systems} for the algebra of the direct product of groupoids.
%%%%
However, it is clear that  the direct sum decomposition of the tensor product algebra given in equation \eqref{eqn: decomposition of the algebra of a classical and quantum system 2} explicitely depends on the fact that the groupoid of the classical system is disconnected, and thus, if we have access only to the information on the algebra of the classical system (i.e, the fact that it is an Abelian algebra), we are not able to immediately detect and appreciate the direct sum decomposition.
%%%%
This instance points to the fact that the groupoid formalism allows a more comprehensive understanding of the properties of physical systems and their associated algebras in a way which is conceptually similar to how the algebraic picture of quantum mechanics provides a more comprehensive understanding of the Hilbert space picture (think, for instance, of the notion of superselection sectors \cite[Chap. IV.1]{Haag-1996}).
%%%%

%%%%

\section{Schr\"{o}dinger cat states again}

We are ready now to analyse Schr\"{o}dinger's {\itshape Gedankenexperiment} in the groupoid formalism.
%%%
Specifically, we consider a composite system $A\times B$, where $A$ is a classical system and $B$ is a quantum system, and a closed evolution of the composite system that reflects the fact that the system does not interact with the external world.
%%%
We need to understand what happens to the composite system when the inital state is a product state in accordance with the fact that we want to model a situation where the   cat and the radioactive material are initially prepared independently and put together ``side by side'' in the box.
%%%

From the point of view of the theory of von Neumann algebras, it is known that the space of states of a von Neumann algebra which is the tensor product of an Abelian algebra with an arbitrary one does not contain entangled states because of Raggio's theorem \cite{Raggio-1981} (see also \cite{Takesaki-2002,Baez-1987,Raggio-1988}).
%%%%
Therefore, if the cat is considered to be a classical system in the sense discussed in the previous sections, then, given  an arbitrary initial product state, it is not possible to create an entangled state of the type given in equation \eqref{eqn: alive-dead entangled state} by means of a dynamical evolution.
%%%%
This implies that there are no inconsistencies between Schr\"{o}dinger's {\itshape Gedankenexperiment}.
%%%
The cat will remain alive (or in the classical state determined by the evolution of the classical system) until the radiative material emits a particle, triggers the detector and the poison is liberated, at which moment the classical state of the cat changes to dead. 
%%%
The state we will find upon opening the box will be $|D\rangle \otimes |-\rangle$, but no  `collapse' of the entangled state vector happened  because there never was an entangled state to begin with.   
%%%

In the rest of this section, we will offer an alternative view on Raggio's theorem based on the groupoid picture introduced before.
%%%%
The main point will be the use of the direct product decomposition of the algebra given in equation \eqref{eqn: decomposition of the algebra of a classical and quantum system 2} and of the properties of the representations of the groupoid $\mathbf{G}_{A}\times\mathbf{G}_{B}$ to show that there is no unitary dynamics of the composite system that can produce an entangled state from a product one.
%%%%
Of course, this result is expected in the light of Raggio's theorem, however, the derivation we offer is completely independent from Raggio's theorem.
%%%%
Moreover, the strategy we adopt to prove our result will allow to present an example in which a conclusion similar to that of Raggio's theorem is true even if both algebras are non-Abelian.
%%%%
Specifically, if one of the two groupoids, say $\mathbf{G}_{A}$, is totally disconnected but its isotropy subgroups are no longer Abelian, then the resulting algebra will be non-Abelian, but there will be a particular family of product states on $\valgab$ that can not evolve into entangled states by means of a unitary dynamics.
%%%%

Given a composite system $A\times B$, where $A$ is a classical system (i.e., it is described by a totally disconnected groupoid $\mathbf{G}_{A}$ with Abelian isotropy subgroups and Abelian von Neumann algebra $\valga$), we know from the discussion in the previous  section that the groupoid $\mathbf{G}_{A\times B}$ of the composite system is the direct product $\mathbf{G}_A\times \mathbf{G}_B$ of the groupoids $\mathbf{G}_{A}$ and $\mathbf{G}_{B}$, and the von Neumann algebra $\valgab$ has the   form written in Eq. \eqref{eqn: decomposition of the algebra of a classical and quantum system 2}.
%%%

%%%%

A closed dynamical evolution  is given by a strongly continuous one-parameter group  $\Phi_{t}$ of automorphism of $\valgab$  \cite{C-I-M-03-2019}.
%%%%
For the purpose of this paper, we consider  dynamical evolutions  $\Phi_{t}$ that can be written in terms of a one-parameter group  $u_{t} = \exp itH$ of unitary elements in $\valgab$ by means of
\be\label{eqn: heisenberg picture of dynamics}
\Phi_{t}(a)\,=\,u_{t}\,a\,u_{t}^{\dagger}\,.
\ee
%%%
This is true if $\mathbf{G}_{B}$ is connected, in which case it gives rise to a non-Abelian von Neumann algebra $\valgb$.
%%%%%
Equation \eqref{eqn: heisenberg picture of dynamics} is the Heisenberg picture of the dynamics, while the Schr\"{o}dinger picture can be obtained by taking the dual action of group $\Phi_{t}$ on the dual of $\valgab$.
%%%
Specifically, if $\rho$ is a state on $\valgab$, then the evolution $\rho_t$ of the state will be given by
\be\label{eqn: schroedinger picture of dynamics}
\rho_{t}(a)\,=\,\rho(\Phi_{t}(a))\,=\,\rho(u_{t}\,a\,u_{t}^{\dagger})
\ee
for all $a\in\valgab$.
%%%%

To meaningfully discuss  Schr\"{o}dinger's {\itshape Gedankenexperiment} in the groupoid picture, we need to find a suitable candidate for the initial state in Eq. \eqref{eqn: schrodinger's cat state}.
%%%%
To build such a state, we need to rely on the so-called  fundamental representation $\pi_{0}$ of the groupoid $\mathbf{G}_{A}\times\mathbf{G}_{B}$  of the system and of its algebra $\valgab$.
%%%%
Here, we limit ourselves to recall only those aspects of the fundamental representation which are strictly needed  in the following, and refer to \cite{C-I-M-02-2019,I-R-2019} for further details.
%%%

The fundamental representation is supported on the separable Hilbert space   $\mathcal{H}_{A\times B} $ given by
\begin{equation}
\mathcal{H}_{A\times B} = L^2(\Omega_A \times \Omega_B) \cong L^2(\Omega_A) \otimes L^2 (\Omega_B) \equiv \mathcal{H}_A\otimes \mathcal{H}_B  \, .
\end{equation}
%%%
Note that   the vectors $|x_{A},x_{B}\rangle = |x_{A}\rangle \otimes |x_{B}\rangle$  corresponding to the delta function at $(x_A,x_B)$ on $\Omega_A \times \Omega_B$  determine an orthonormal basis of $\mathcal{H}_{A\times B} $.
%%%
The representation of a transition $\alpha\colon(y_{A},y_{B})\lra(z_{A},z_{B})$ in $\mathbf{G}_{A}\times\mathbf{G}_{B}$ on a basis element $|x_{A},x_{B}\rangle$ is given by
\be\label{eqn: fundamental representation of composite algebra}
\pi_{0}(\alpha)|x_{A},x_{B}\rangle\,=\,\delta_{x_{A}y_{A}}\,\delta_{x_{B}y_{B}}|z_{A},z_{B}\rangle .
\ee
%%%%
Given an element $a=\sum_{j}a^{j}\alpha_{j}$ in $\valgab$ (where the sum may be infinite provided it converges in the weak topology), we have
\be\label{eqn: fundamental representation of composite algebra 2}
\pi_{0}(a)|x_{A},x_{B}\rangle\,=\,\sum_{j}\,a^{j}\pi_{0}(\alpha_{j})|x_{A},x_{B}\rangle\,.
\ee
%%%%
%%%%
%Referring to the decomposition of $a$ given in equation \eqref{eqn: decomposition of elements the algebra of a classical and quantum system}, it is then clear that equation \eqref{eqn: fundamental representation of composite algebra} and \eqref{eqn: fundamental representation of composite algebra 2} imply the existence of a complex number $a'$ such that, for every   $a_{x_{A}}$ $a_{x_{A}}'$ in $\valgxa$, it is
%\be
%\pi_{0}(a_{x_{A}}'\otimes a_{B})\,=\,a'\pi_{0}(a_{x_{A}}\otimes a_{B}) 
%\ee
%because every elementary transition in the isotropy algebra $\valgxa$ acts trivially on $\mathcal{H}_{A\times B}$.
%%%

On the other hand, the Hilbert space $\mathcal{H}_{A\times B}$ decomposes as the direct integral 
\begin{equation}
\mathcal{H}_{A\times B}= \int^\oplus_{\Omega_A} \mathcal{H}_{x_A}\, d\mu (x_A) \, ,
\end{equation}
where every 
\be\label{eqn: invariant subspace of fundamental representation}
\mathcal{H}_{x_A}= \mathbb{C}|x_A\rangle \otimes \mathcal{H}_B\,
\ee
is an invariant subspace, that is, we have $\pi_0(a)\mathcal{H}_{x_A}\subset \mathcal{H}_{x_A}$ for every $a \in \valgab$   (note also that $\mathcal{H}_{x_{A}}\,\cong\,\mathcal{H}_{B}$).  
%%%
Actually, for any finite subset $\Delta \subset \Omega_A$, we get the (finite-dimensional) subspace $\mathcal{H}_\Delta \subset \mathcal{H}_{A\times B}$ given by
\be
\mathcal{H}_\Delta = \bigoplus_{\substack{x_A\in \Delta}}\mathcal{H}_{x_A}\,.
\ee
%%%
Since $\mathbf{G}_A$ is totally disconnected ($A$ is a classical system), Eq. \eqref{eqn: fundamental representation of composite algebra 2} leads to
\be
\pi_{0}(a)|x_{A},x_{B}\rangle = \sum_j a^j |x_{A},y_{Bj}\rangle \, , 
\ee
where $a = \sum_j a^j \alpha_j$ with $\alpha_j = \gamma_{Aj} \otimes \alpha_{Bj}$,   $\gamma_{Aj} \in \mathbf{G}_{x_A}$, and $\alpha_{Bj} \colon x_B \to y_{Bj}$, and to 
\be
\pi_{0}(a)|x_{A},x_{B}\rangle = \mathbf{0}  
\ee
if $\gamma_{Aj}\in\mathbf{G}_{x_{Aj}}$ with $x_{Aj}\neq x_{A}$ for all $j$.
%%%

This means that the fundamental representation $\pi_0$ has a block diagonal structure  graphically depicted as
\begin{equation} \label{eqn: decomposition of fundamental representation}
\pi_{0} = 
\setlength{\fboxsep}{0pt}%\setlength{\arraycolsep}{0pt}
\left(\begin{array}{@{}c@{}c@{}c@{\mkern-5mu}c@{\,}cc*{2}{@{\;}c}@{}}%
\noalign{\vskip 1.5ex}
 \ddots \\
  &\fbox{\,$\begin{matrix}
 \pi_{x_A} \\
  \end{matrix}$\,}
 \\[-0.4pt]
  && \hskip-0.4pt\fbox{\,$\begin{matrix}
  \pi_{x_A'} \\ 
  \end{matrix}$\,}\\[-0.5pt]
   &&& \hskip-0.5pt \fbox{\,$\begin{matrix}
 \pi_{x_A''} \\
  \end{matrix}$\,}
 \\[-0.4pt]
   &&&& \;\ddots
\end{array}\right),
\end{equation}
where $\pi_{x_A} = \pi_0|_{\mathcal{H}_{x_A}}$, or, in other words: $\pi_0 = \bigoplus_{x_A \in \Omega_A} \pi_{x_A}$.
%%%

Now, we are ready to discuss Schr\"{o}dinger's {\itshape Gedankenexperiment}.
%%%
At this purpose, we consider the normalized vector $| \Psi_{x_{a}} \rangle \in\mathcal{H}_{A\times B}$ given by
\be 
|\Psi_{x_{a}}\rangle\,=\,|x_{A},\psi_{B}\rangle\,=\,|x_{A}\rangle\,\otimes\,|\psi_{B}\rangle\,.
\ee
with $|\psi_B\rangle \in \mathcal{H}_B$.
%%%
This vector defines a state $\rho_{x_{A}}$ on $\valgab$ by means of
\be\label{eqn: schrodinger cat states}
\rho_{x_{A}}(a)\,:=\,\langle \Psi_{x_{a}}|\pi_{0}(a)|\Psi_{x_{a}}\rangle,
\ee
and $\rho_{x_{A}}$ is clearly a product state on $\valgab$.
%%%%
The state $\rho_{x_{A}}$ is the analogue of the initial state of the Schr\"{o}dinger's cat considered in the introduction section \ref{sec.1:Introduction}.
%%%%
We will now analyze its dynamical evolution $\Phi_{t}^{\sharp}(\rho_{x_{A}})$ and show that it is a product state for every $t$.
%%%

Indeed, the Heisenberg evolution represented in $\mathcal{H}_{A\times B}$ takes the form
\be
\pi_{0}\left(\Phi_{t}(a)\right)\,=\,\pi_{0}\left(u_{t}\,a\,u_{t}^{\dagger}\right)\,=\,\pi_{0}(u_{t})\,\pi_{0}(a)\,\pi_{0}(u_{t}^{\dagger}),
\ee
and we immediately conclude that
\be
\left(\Phi_{t}^{\sharp}(\rho_{x_{A} })\right)(a)\,=\,\langle \Psi_{x_{A} }|\pi_{0}(u_{t})\,\pi_{0}(a)\,\pi_{0}(u_{t}^{\dagger})|\Psi_{x_{A} }\rangle\,.
\ee
%%%%
Then, recalling that the subspace given in equation \eqref{eqn: invariant subspace of fundamental representation} is invariant with respect to $\pi_{0}$, we have that
\be
\pi_{0}(u_{t}^{\dagger})|\psi\rangle\,=\,\pi_{0}(u_{t}^{\dagger})\left(|x_{A}\rangle\otimes|\psi_{B}\rangle\right)\,=\,|x_{A}\rangle\otimes|\psi_{B}^{t}\rangle,
\ee
which means that $\rho_{x_{A} }=\Phi_{t}^{\sharp}(\rho_{x_{A} })$ is a product state for all $t$ as claimed.
%%%
Furthermore, we may consider separable states of the form
\be\label{eqn: schrodinger cat states 2}
\varrho\,=\,\sum_{x_{A}\in\Omega_{A}}\,p_{x_{A}}\,\rho_{x_{A} }
\ee
with $\rho_{x_{A} }$ as in equation \eqref{eqn: schrodinger cat states}, and with  $p_{x_{A}}\geq 0$ such that $\sum_{x_{A}\in\Omega_{A}}p_{x_{A}}=1$, and we immediately obtain that $\varrho_{t}=\Phi_{t}^{\sharp}(\varrho)$ is a separable state for all $t$.
%%%
Note that the same statement holds if we start with a mixed state both for the quantum system or the classical system.

The results of this section are summarized in the following theorem:

\begin{theorem}\label{thm: schroedinger's cat}
Let A be a  system described by a totally disconnected, countable groupoid $\mathbf{G}_{A}\rightrightarrows\Omega_{A}$.
%%%
Let $B$ be a system described by a countable, connected groupoid $\mathbf{G}_{B}\rightrightarrows\Omega_{B}$.
%%%
Let $A\times B$ be the composite system described by the direct product groupoid  $\mathbf{G}_A\times \mathbf{G}_B \rightrightarrows\Omega_{A}\times\Omega_{B}$.
%%%
Given a dynamical evolution of $A\times B$ described by a one-parameter group $\Phi_{t}$ of automorphisms of the von Neumann algebra $\valgab$ of the system that can be written according to Eq. \eqref{eqn: heisenberg picture of dynamics}, assume that at the initial time   $t_0$  the system is in the separable state $\varrho$ given by Eq. \eqref{eqn: schrodinger cat states 2}.
%%%
Then, the dynamical evolution of $\varrho$ will be a separable state   for all $t$. 
\end{theorem}

As a consequence, referring to the cat states of section \ref{sec.1:Introduction}, in the fundamental representation of the groupoid algebra, it is not possible to construct entangled states of the form:
\begin{equation}
|\psi\rangle = \left( |D\rangle \otimes|+\rangle + |A\rangle \otimes |-\rangle \right)
\end{equation}
starting with a state of the form $|\psi_0\rangle = |A\rangle \otimes|-\rangle$ and evolving with a unitary evolution inside the algebra.   
%%%
%%%
It is important to note that this result depends crucially on the form of the groupoid algebra of a composite system made of a classical and a (possibly) quantum systems, and of its fundamental representation.
%%%
Indeed, we see that what is actually needed in the argument outlined above is the fact that the subspace given in equation \eqref{eqn: invariant subspace of fundamental representation} is invariant with respect to the fundamental representation $\pi_{0}$.
%%%
This instance is realized even if the groupoid $\mathbf{G}_{A}$ is a countable groupoid which is totally disconnected but presents non-Abelian isotropy subgroups.
%%%%
In this case, the algebra $\valga$ is no longer Abelian (the system $A$ is thus non-classical) but the argument given above still applies.
%%%%

\section{Conclusions}

In this contribution, we exploited Schr\"{o}dinger's famous {\itshape Gedankenexperiment} with cats as a departure point to analyse the issue of composition of classical and quantum systems in the groupoid formalism introduced recently.
%%%
The first step was to characterize what is a classical system in the groupoid formalism mentioned before.
%%%
In the context of countable groupoids, we found that  a classical system in the sense of the quantum logic of  Birkhoff and von Neumann is necessarily associated with a totally disconnected groupoids with Abelian isotropy subgroups, and viceversa.
%%%%
Then, we asked what happens to a product state of a composite system made up of a classical system and a quantum one.
%%%%
The arguments presented in this paper show that, under the appropriate technical conditions expressed in Thm. \ref{thm: schroedinger's cat}, the direct composition of a classical system with a quantum one prevents a family of separable states of the composite systems  to evolve into entangled states by means of unitary evolutions.  
%%%
As said before, if the system $A$ is classical in the sense specified in subsection \ref{sec:classical}, the content of Theorem \ref{thm: schroedinger's cat} is a particular case of Raggio's theorem.
%%%
However, when the isotropy subgroups of $A$ are non-Abelian, then the algebra $\valga$ is non-Abelian and Raggio's theorem does not apply, while Theorem \ref{thm: schroedinger's cat} does apply, and allows us to conclude that there are separable states of the composite system that will remain separable under all possible unitary dynamics of the composite system.
%%%%
Therefore, even if the cat is non-classical, there still are particular inital separable states of the composite system for which Schr\"{o}dinger's disturbing evolution is precluded.
%%%
These type of states may be called ``pseudo-classical'' states, and the results of this section  may be a  hint to pursue a more general investigation of  pseudo-classical states of composite systems.
%%%
In this direction, Theorem \ref{thm: schroedinger's cat} represents a first step stating clearly that pseudo-classical states actually exist.
%%%
In the context of quantum information theory, a similar results has already been obtained \cite{G-B-2003}.
%%%
In this case, both system A and B are described by connected groupoids (the groupoids of pair of n elements giving rise to matrix algebras), and it is proved that there are mixed separable states that remain separable under every unitary dynamics.
%%%
The set of such states (which are pseudo-classical in the sense specified above) is referred to as the Gurvits-Barnum ball, and it would be interesting to understand the proof given in \cite{G-B-2003} in the context of the groupoid picture.
%%%
All these issues will be dealt with in future works.

\addcontentsline{toc}{section}{References}
%\bibliographystyle{plain}
%\bibliography{scientific_bibliography}

\begin{thebibliography}{10}

\bibitem{Auletta-2001}
Auletta, G.
\newblock {\em Foundations and Interpretations of Quantum Mechanics}; World Scientific Co.,  2001.

\bibitem{Baez-1987}
Baez, J.
\newblock {Bell's Inequality for $C^{*}$-Algebras}.
\newblock {\em Letters in Mathematical Physics}
  {\bf 1987}, {\em 13}, 135--136.



\bibitem{B-vN-1936}
Birkhoff, G.; J., V.N.
\newblock The Logic of Quantum Mechanics.
\newblock {\em Annals of Mathematics} {\bf 1936}, {\em 37},~823--843.

\bibitem{C-G-2019}
Chitambar, E.; Gour, G.
\newblock {Quantum resource theories}.
\newblock {\em Reviews of Modern Physics} {\bf 2019}, {\em 91(2)}, 025001(48).

\bibitem{C-DC-I-M-02-2020}
F.~M. Ciaglia, F.~Di~Cosmo, A.~Ibort, and G.~Marmo.
\newblock {Schwinger's Picture of Quantum Mechanics}.
\newblock {\em International Journal of Geometric Methods in Modern Physics},
  17(04):2050054 (14), 2020.

\bibitem{C-DC-I-M-2020}
F.~M. Ciaglia, F.~Di~Cosmo, A.~Ibort, and G.~Marmo.
\newblock {Schwinger's Picture of Quantum Mechanics IV: Composition and
  independence}.
\newblock {\em International Journal of Geometric Methods in Modern Physics},
  17(04):2050058 (34), 2020.

\bibitem{C-I-M-2018}
F.~M. Ciaglia, A.~Ibort, and G.~Marmo.
\newblock {A gentle introduction to Schwinger's formulation of quantum
  mechanics: the groupoid picture}.
\newblock {\em Modern Physics Letters A}, 33(20):1850122--8, 2018.

\bibitem{C-I-M-02-2019}
F.~M. Ciaglia, A.~Ibort, and G.~Marmo.
\newblock {Schwinger's Picture of Quantum Mechanics I: Groupoids}.
\newblock {\em International Journal of Geometric Methods in Modern Physics},
  16(08):1950119 (31), 2019.

\bibitem{C-I-M-03-2019}
F.~M. Ciaglia, A.~Ibort, and G.~Marmo.
\newblock {Schwinger's Picture of Quantum Mechanics II: Algebras and
  Observables}.
\newblock {\em International Journal of Geometric Methods in Modern Physics},
  16(09):1950136 (32), 2019.

\bibitem{C-I-M-05-2019}
F.~M. Ciaglia, A.~Ibort, and G.~Marmo.
\newblock {Schwinger's Picture of Quantum Mechanics III: The Statistical
  Interpretation}.
\newblock {\em International Journal of Geometric Methods in Modern Physics},
  16(11), 2019.
  
  \bibitem{Connes-1994}
Connes, A.
\newblock {\em Noncommutative geometry}; Academic Press,  1994.
  
\bibitem{E-P-R-1935} Einstein, A.; Podolsky, B.; Rosen, N. 
\newblock {Can Quantum-Mechanical Descripton of Physical Reality be Considered Complete?}
\newblock {\em Physical Review}
  {\bf 1935}, {\em 47}, 777-780.

\bibitem{G-B-2003}
L.~Gurvits and H.~Barnum.
\newblock {Separable balls around the maximally mixed multipartite quantum
  state}.
\newblock {\em Physical Review A}, 68(4), 2003.

\bibitem{Haag-1996}
R.~Haag.
\newblock {\em Local quantum physics: Fields, particles, algebras}.
\newblock Springer-Verlag, Berlin, 1996.

\bibitem{I-R-2019}
A.~Ibort and M.~A. Rodriguez.
\newblock {\em {An introduction to the theory of groups, groupoids and their
  representations}}.
\newblock CRC, 2019.

\bibitem{Jauch-1952}
J.~M. Jauch.
\newblock {\em Foundations of Quantum Mechanics}.
\newblock Addison-Wesley, Reading, MA, 1952.

\bibitem{L-Y-1999}
Lieb, E.H.; Yngvason, J.
\newblock The physics and mathematics of the second law of thermodynamics.
\newblock {\em Physics Reports} {\bf 1999}, {\em 310},~1--96.

\bibitem{L-Y-2000}
Lieb, E.H.; Yngvason, J., Visions in Mathematics; Birk\"{a}user Verlag, Basel,
  2000; chapter The Mathematics of the Second Law of Thermodynamics, pp.
  334--358.

\bibitem{M-R-2007}
Marchetti, P.A.; Rubele, R.
\newblock Quantum Logic and Non-Commutative Geometry.
\newblock {\em International Journal of Theoretical Physics} {\bf 2007}, {\em  46},~49--62.

\bibitem{Maudlin-1995} 
Maudlin, T.
\newblock {Three Measurement Problems}.
\newblock{\em Topoi} {\bf 1995}, {\em 14}, 7--15.

\bibitem{Raggio-1981}
Raggio, G. A. 
\newblock{States and Composite Systems in $W^{*}$-algebraic Quantum Mechanics}.
\newblock{\em {Ph.D. Thesis ETH Zurich, No. 6824, ADAG Administration $\&$ Druck AG, Zurich,} 1981.}

\bibitem{Raggio-1988}
Raggio, G. A.
\newblock {A Remark on Bell's Inequality and Decomposable Normal States}.
\newblock {\em Letters in Mathematical Physics}
  {\bf 1988}, {\em 15}, 27--29.


\bibitem{Renault-1980} Renault, J. 
\newblock {{\em A Groupoid Approach to $C^{\star }-$Algebras}}; Springer-Verlag, Berlin 1980.

\bibitem{Schrodinger-1935}
E.~Schr\"{o}dinger.
\newblock {Die gegenw\"{a}rtige Situation in der Quantenmechanik }.
\newblock {\em Die Naturwissenschaften}, 48:807 -- 812, 1935.

\bibitem{Schwinger-2000}
J.~Schwinger.
\newblock {\em {Quantum Kinematics and Dynamics}}.
\newblock Westview Press, 2000.

\bibitem{Sc51-1} 
 Schwinger, J.
 \newblock The Theory of Quantized Fields I.
 \newblock {\em  Physical Reviews} {\bf 1951}, {\em 82(6)}, 914--927.

\bibitem{Sc51-2} 
 Schwinger, J.
 \newblock The Theory of Quantized Fields II.
 \newblock {\em  Physical Reviews} {\bf 1953}, {\em 91(3)}, 713--728.
 
 \bibitem{Sc51-3} 
 Schwinger, J.
 \newblock The Theory of Quantized Fields III.
 \newblock {\em  Physical Reviews} {\bf 1953}, {\em 91(3)}, 728--740.

 \bibitem{Sc51-4} 
 Schwinger, J.
 \newblock The Theory of Quantized Fields IV.
 \newblock {\em  Physical Reviews} {\bf 1953}, {\em 92(5)},  1283--1300.

 \bibitem{Sc51-5} 
 Schwinger, J.
 \newblock The Theory of Quantized Fields V.
 \newblock {\em  Physical Reviews} {\bf 1954}, {\em 93(3)},  615--624.

 \bibitem{Sc51-6} 
 Schwinger, J.
 \newblock The Theory of Quantized Fields VI.
 \newblock {\em  Physical Reviews} {\bf 1954}, {\em 94(5)}, 1362--1384.
 

\bibitem{Sorkin-1994}
R.~D. Sorkin.
\newblock Role of time in the sum-over-histories framework for gravity.
\newblock {\em {International Journal of Theoretical Physics}}, 33(3):523 --
  534, 1994.

\bibitem{Takesaki-2002}
M.~Takesaki.
\newblock {\em {Theory of Operator Algebra I}}.
\newblock Springer-Verlag, Berlin, 2002.

\bibitem{von-Neumann-1927}
von Neumann, J.
\newblock Wahrscheinlichkeitstheoretischer Aufbau der Quantenmechanik.
\newblock {\em Nachrichten von der Gesellschaft der Wissenschaften zu
  G\"ottingen, Mathematisch-Physikalische Klasse} {\bf 1927}, pp. 245--272.

\bibitem{von-Neumann-1955}
J.~von Neumann.
\newblock {\em {Mathematical Foundations of Quantum Mechanics}}.
\newblock Princeton University Press, Princeton, NJ, 1955.

\end{thebibliography}

\end{document}